\pdfoutput=1
\documentclass{article}
\usepackage{amsthm}
\usepackage[utf8]{inputenc}
\usepackage{thmtools}
\usepackage{thm-restate}
\usepackage{amssymb}
\usepackage{amsmath}
\usepackage[inline]{enumitem}
\usepackage{hyperref}
\usepackage{xcolor}
\usepackage[title]{appendix}

\usepackage[affil-it]{authblk}

\hypersetup{colorlinks,linkcolor={red!70!black},citecolor={blue!70!black}}
\newcommand\opn[1]{\operatorname{#1}}
\newcommand\NN{\mathbb N}
\newcommand\mc[1]{\mathcal{#1}}
\renewcommand\H{\mc H}
\newcommand\lft{1}
\newcommand\rgt{2}
\newcommand\Hl{\mc H_{\lft}}
\newcommand\Hr{\mc H_{\rgt}}

\newcommand\bs[1]{\boldsymbol{#1}}
\newcommand\iiff{\quad\iff\quad}

\newcommand\lift{\opn{lift}}
\newcommand\liftofrom[2]{\lift_{#1\ot#2}}

\newcommand\id{\mathbf{I}}

\newcommand\other{\text{literature}}
\newcommand\ot{\leftarrow}
\newcommand\into\hookrightarrow

\newcommand\GS{\mc Z}
\newcommand\Z{\GS}
\newcommand\double[1]{\mathbb{#1}}
\newcommand\CC{\double{C}}
\newcommand\V{\mc V}

\newcommand\Y{\mc Y}

\newcommand\onto[2]{{\mu(\tfrac{#1}{#2})}}

\newcommand\spec{\opn{spec}}

\newcommand\subsp{\preceq}
\newcommand\supsp{\succeq}
\newcommand\opge\supsp
\newcommand\erra{\varphi}
\newcommand\ople\subsp

\newcommand\plusbounded{\lesssim}
\newcommand\aand{\quad\text{and}\quad}

\newcommand\sphere{\mc S}
\newcommand\vdim[1]{\opn{dim}(#1)}

\newcommand\ONE{1\kern-1.0ex\opn{I}}
\newcommand\shrink{\Delta}
\renewcommand\P{\opn{P}}
\newcommand\nofrac[2]{#1\big/#2}
\newcommand\xpt\backslash
\newcommand\IN[1]{\Gamma_{#1}^\dag}
\newcommand\ON[1]{\Gamma_{#1}}
\newcommand\tofrom[2]{\Pi_{#1\ot\!#2}}
\newcommand\W{\mc W}
\newcommand\B{\mathcal B}

\newcommand\AGSP{{K}}
\newcommand\agsp{{K}}

\newcommand\emp[1]{\textbf{\textit{#1}}}
\newcommand\Bl{\B(\Hl)}
\newcommand\Br{\B(\Hr)}
\newcommand\PAP{\mc K}

\newcommand\bra[1]{\langle#1|}
\newcommand\shann{\opn{S}}

\newcommand\ket[1]{|#1\rangle}
\newcommand\bracket[2]{\langle#1|#2\rangle}
\newtheorem{lemma}{Lemma}
\numberwithin{lemma}{section}
\newtheorem{theorem}[lemma]{Theorem}
\newtheorem*{lemma*}{Lemma}
\newtheorem{proposition}[lemma]{Proposition}
\newtheorem{observation}[lemma]{Observation}
\newtheorem{claim}[lemma]{Claim}
\newtheorem{definition}[lemma]{Definition}
\newtheorem*{definition*}{Definition}
\newtheorem{corollary}[lemma]{Corollary}

\author{Nilin Abrahamsen}

\title{Sharp implications of AGSPs for degenerate ground spaces}

\begin{document}

\affil{\small{Department of Mathematics,\\Massachusetts Institute of Technology,\\Cambridge, MA, USA}}

\maketitle

\begin{abstract}We generalize the `off-the-rack' \emph{AGSP$\Rightarrow$entanglement bound} implication of [Arad, Landau, and Vazirani '12] from unique ground states to degenerate ground spaces. Our condition $R\shrink\le1/2$ on a $(\shrink,R)$-AGSP matches the non-degenerate case, whereas existing tools in the literature of spin chains would only be adequate to prove a less natural implication which assumes $R^{\text{Const}}\shrink\le c$. To show that $R\shrink\le1/2$ still suffices in the degenerate case we prove an optimal error reduction bound which improves on the literature by a factor $\delta\mu$ where $\delta=1-\mu$ is the viability. 

The generalized off-the-rack bound implies the generalization of a recent 2D subvolume law of [Anshu, Arad, and Gosset '19] from the non-degenerate case to the sub-exponentially degenerate case.\end{abstract}

\section{Introduction}
Approximate ground space projectors (AGSPs) are an indispensable tool for proving entanglement bounds on ground states of gapped local Hamiltonians \cite{arad_improved_2012,arad_area_2013,anshu_entanglement_2019} and for constructing polynomial-time algorithms \cite{arad_rigorous_2017} for 1D Hamiltonians. In such results it is often assumed that the ground state be \emph{unique} \cite{arad_improved_2012,arad_area_2013,anshu_entanglement_2019}. A main reason for the ubiquitousness of this assumption is that for unique ground states, the existence of a $(\shrink\!=\!\frac1{2R},\:R)$-AGSP implies a readymade entanglement bound $O(\log R)$ by a lemma of Arad, Landau, and Vazirani (\cite{arad_improved_2012} corollary III.4). This fact, which we call the \emph{off-the-rack}\footnote{By `off-the-rack' we mean that the bound follows from a single condition and does not require a case-specific analysis outside of verifying this condition.} (OTR) bound, reduces the task of proving an area law to that of constructing such an AGSP in the non-degenerate setting.

Generalizing entanglement bounds and algorithms for 1D gapped Hamiltonians from the setting of a unique ground state to a ground space with \emph{degeneracy} (i.e., dimension) $D>1$ has been a main focus of several works, starting with the case of a \emph{constant} degeneracy \cite{chubb_computing_2016,huang_area_2014} and later generalized further to \emph{polynomial} degeneracy \cite{arad_rigorous_2017}.

While AGSPs have been used before to prove a 1D area law for polynomially degenerate ground spaces \cite{arad_rigorous_2017}, no direct analogue of the OTR bound follows using only existing tools and analyses. Indeed, inspecting the state-of-the-art proofs of degenerate-case entanglement bounds one finds the necessary assumption on a $(\shrink,R)$-AGSP to be $R^C\shrink\le1/2$ (where one can take $C=12$ \cite{arad_rigorous_2017}). Here $R$ is the entanglement rank of the AGSP and $\shrink$ the shrinking factor. This discrepancy with the non-degenerate case is somewhat unsatisfactory\footnote{A typical analysis, say, for 1D area laws, goes by showing $\nofrac{\log(\shrink^{-1})}{\log R}\to\infty$ which implies a bound on $R^C\shrink$ just as it does on $R\shrink$. Thus, our focus on ensuring that $R$ appears with exponent $1$ in proposition \ref{mainres} may be said to be of little practical consequence. And indeed we pursue this goal mainly for aesthetics.}, and in particular a bound on $R^C\shrink$ does not follow from the OTR assumption $R\shrink\le 1/2$ by any amplification procedure (amplification instead gives a bound on $R^C\shrink^C$).

\subsection{Our contribution}
We generalize the off-the-rack entanglement bound of \cite{arad_improved_2012} to degenerate ground spaces with no strengthening of the assumed parameter tradeoff:

\begin{proposition}\label{mainres}
	Suppose there exists an $(\shrink,R)$-AGSP $\AGSP\in\Hl\otimes\Hr$ such that
	\[R\shrink\le1/2.\]
	Let $\GS$ be the target space of $\AGSP$ and $D=\vdim\GS$ its degeneracy. Then the maximum entanglement entropy of any state $\ket\psi\in\GS$ satisfies the bound
	\[\max_{\ket\psi\in\sphere(\GS)}\shann(\rho^\psi_\lft)=1.01\cdot\log D+O(\log R),\]
	where $\sphere(\GS)$ is the set of unit vectors in $\GS$ and $\shann(\rho^\psi_\lft)$ is the entanglement entropy of $\ket\psi$ between subsystems $\Hl$ and $\Hr$.
\end{proposition}

\paragraph{Approximate target spaces for frustrated Hamiltonians.} In the case of a frustrated Hamiltonian the typical AGSP contruction involves spectral \emph{truncations} of parts of the Hamiltonian, incurring an error in the target space of the AGSP. We therefore also prove a version (lemma \ref{frversion}) of proposition \ref{mainres} which is applicable to the frustrated case by allowing the target space to be approximate.

To obtain proposition \ref{mainres} we prove the {optimal bound} $\frac{\delta'}{\mu'}\le\shrink\frac\delta\mu$ on the error reduction from an AGSP (lemma \ref{sharplemma}), improving the best bound in the literature by a factor $\delta\mu$. Here $\mu$ represents the \emph{overlap} and $\delta$ the \emph{viability error} \cite{arad_rigorous_2017}, and symbols with a prime correspond to parameters after applying the AGSP. The improved error reduction bound is essential when we apply the \emph{bootstrapping} argument \cite{arad_improved_2012,arad_area_2013,arad_rigorous_2017} to finish the proof of the entanglement bound.

\subsection{Sharp error reduction bound}
Given a Hilbert space $\H$ consider two subspaces $\GS,\V\subsp\H$ such that $\V$ \emph{covers} $\GS$, meaning the projection $\P_\GS$ onto $\GS$ is surjective even when restricted to $\V$, $\P_\GS(\V)=\GS$. Denoting the largest \emph{principal angle} \cite{galantai_jordans_2006,ben-israel_geometry_1967} between $\GS$ and $\P_\V(\GS)\subsp\V$ as $\theta$ we define the \emp{error ratio} $\erra$ of $\V$ onto $\GS$ as
\[\erra:=\tan^2\theta.\]
The \emph{overlap} $\mu=\min_{\ket z\in\sphere(\Z)}\bra z\P_\V\ket z$ and \emph{viability} error $\delta=1-\mu$ \cite{arad_rigorous_2017} of $\V$ onto $\Z$ coincide with $\mu=\cos^2\theta$ and $\delta=\sin^2\theta$, so the error ratio is equivalently characterized as $\delta/\mu$. 
\begin{lemma}[Sharp error reduction]\label{sharplemma}
Let $\AGSP$ be a $\shrink$-AGSP for $\GS\subsp\H$, and suppose $\V\subsp\H$ covers $\GS$ with error ratio $\erra$. then $\V':=\AGSP(\V)=\{\AGSP\ket v:\ket v\in\V\}$ covers $\GS$ with and the error ratio $\erra'$ of $\V'$ onto $\GS$ satisfies
\[\erra'\le\shrink\cdot\erra.\]
\end{lemma}
This bound is clearly sharp\footnote{Consider the $\shrink$-AGSP $\AGSP=\ket0\bra0+\sqrt\shrink\ket1\bra1$ on $\CC^2$ and subspaces $\Z,\V\subsp\CC^2$ spanned by $\ket z=\ket0$ and $\ket v=\frac1{\sqrt{1+\erra}}(\ket 0+\sqrt\erra\ket 1)$.}. 
I sketched a proof of lemma \ref{sharplemma} in a restricted formulation in (\cite{abrahamsen_polynomial-time_2019} observation 8.5). Here we give the full proof in section \ref{thesharpproof} and generalize the bound to hold for any typical definition of an AGSP. The proof is based on switching the roles of $\Z$ and $\V$ using the \emph{symmetry lemma} \cite{abrahamsen_polynomial-time_2019}.

In section \ref{liftingversion} we also include an alternative proof of lemma \ref{sharplemma} which is more similar in structure to the proof of a weaker bound in \cite{arad_rigorous_2017} lemma 6. In this case we obtain the strengthened bound by improving lemmas (1 and 2)\footnote{We will refer to these two lemmas, and our improved version, as `lifting lemmas'.} of \cite{arad_rigorous_2017} to have quadratically better dependence on the overlap $\mu$.

Because $\delta'=\frac{\erra'}{1+\erra'}\le\erra'$, lemma \ref{sharplemma} implies:
\begin{corollary}\label{slightly}
If $\V$ covers $\GS$ with overlap $\mu=1-\delta$, then $\AGSP\V$ covers $\GS$ with viability error $\delta'\le \shrink\delta/\mu$.
\end{corollary}

The previous state-of-the-art error reduction bound for the general degenerate-case AGSPs (\cite{arad_rigorous_2017} lemma 6) bounded the post-AGSP viability error by
\[\delta_{\other}'=\shrink/\mu^2.\]
The post-AGSP error bound $\delta'$ from corollary \ref{slightly} improves on this bound by a factor $\mu\cdot\delta$, 
which is particularly significant when starting in either the small-overlap $\mu\ll1$ or small-error regime $\delta\ll1$. 

\subsection{Consequences}
The analysis in this work implies that results on local Hamiltonians for unique ground states can be straightforwardly extended to the degenerate case. As an illustration of this we consider an important recent advance in the understanding of 2D spin systems \cite{anshu_entanglement_2019}:
\subsubsection{Generalization of the locally-gapped 2D subvolume law \cite{anshu_entanglement_2019}}
A recent advance of Anshu, Arad, and Gosset \cite{anshu_entanglement_2019} proved a subvolume law for the unique ground state of a frustration-free local Hamiltonian on a 2D lattice in terms of the \emph{local gap} $\gamma$, i.e., the smallest gap of a subsystem. This quantity is motivated, e.g., by \emph{finite-size criteria} \cite{knabe_energy_1988,gosset_local_2016,lemm_finite-size_2019}; we refer to \cite{anshu_entanglement_2019} for details. This represented significant progress in understanding the entanglement structure of local Hamiltonian systems with gap conditions in 2D, providing evidence in favor of the conjectured \emph{area law}. The theorem states (slightly paraphrased):
\begin{theorem}[\cite{anshu_entanglement_2019}]\label{AAGthm}
	Let $H$ be a frustration-free Hamiltonian on a $L_1\times L_2$ lattice of qudits (each with Hilbert space $\CC^d$) and local gap $\gamma$.  If the ground state $\ket\psi$ of $H$ is unique, then the entanglement entropy $\shann(\rho^\psi_{\text{left}})$ of $\ket\psi$ across a vertical cut (of height $L_2$) satisfies
	\[\shann(\rho^\psi_{\text{left}})=O\Big(\big(L_2/\!\sqrt\gamma\big)^{\frac53}\log^{\frac73}(\tfrac{dL_2}\gamma)\Big).\]
\end{theorem}
Inspecting the proof by \cite{anshu_entanglement_2019} it is clear (see appendix \ref{2d_gen}) that the analysis of the shrinking and entanglement parameters of their AGSP do not depend on the degeneracy of the ground space. Applying proposition \ref{mainres} to the AGSP of \cite{anshu_entanglement_2019} one obtains:
\begin{corollary}
	Let $H$ be a 2D lattice Hamiltonian satisifying the conditions of theorem \ref{AAGthm}, except now allow the ground space $\GS=\opn{Ker} H$ to have arbitrary dimension $D=\vdim\GS$. Then,
	\[\max_{\ket\psi\in\sphere(\GS)}\shann(\rho^\psi_{\text{left}})=O\Big(\big(L_2/\!\sqrt\gamma\big)^{\frac53}\log^{\frac73}(\tfrac{dL_2}\gamma)+\log D\Big),\]
	where $\shann(\rho^\psi_{\text{left}})$ is the entanglement entropy of $\ket\psi$ across an arbitrary vertical cut.
\end{corollary}
In particular,
	\begin{itemize}
		\item The entanglement bound is the same as \cite{anshu_entanglement_2019} up to a constant factor when the degeneracy has growth at most $D=2^{O(L_2^{5/3})}$.
		\item In the parameter regime for $L_1,L_2,\gamma$ where \cite{anshu_entanglement_2019} yields a subvolume law (e.g., $L_1=L_2$ and $\gamma=\Omega(1)$), a subvolume law still holds for \emph{any} sub-exponential degeneracy $D=2^{o(L_1L_2)}$.
	\end{itemize}

\section{Preliminaries}

Given a Hamiltonian $H$, an AGSP (approximate ground space projector) for $H$ is an operator $\AGSP$ which shrinks the excited states but not the vectors in the ground space $\GS$ of $H$. We do not directly invoke the Hamiltonian itself, as the AGSP property (but generally not the contruction of an AGSP) can be captured in terms of just the ground space $\GS$. 

In the interest of broad applicability we define an AGSP such that the definitions used in the literature \cite{arad_improved_2012,arad_area_2013,arad_rigorous_2017} are all special cases of the definition used here. 

\subsection{AGSPs}
Let $\B(\H)$ denote the space of all bounded linear operators on Hilbert space $\H$.
\begin{definition}\label{AGSPdef}
	A $\shrink$-AGSP with target space $\Z\subsp\H$ is an operator $\AGSP\in\B(\H)$ which commutes with $\P_\Z$ and satisfies
	\begin{enumerate}\item
		$\P_\Z \AGSP^\dag\AGSP\P_\Z\opge\P_\Z$,\qquad i.e., $\AGSP$ is a dilation on $\GS$\label{dilation}
	\item$\|\AGSP\P_{\Z^\perp}\|\le\sqrt\shrink$.\label{shrinkitem}
	\end{enumerate}
A $(\shrink,R)$-AGSP is a $\shrink$-AGSP $\AGSP\in\Bl\otimes\Br$ with entanglement rank at most $R$.
\end{definition}
The condition that $\AGSP$ commute with $\P_\Z$ is equivalent with the following two conditions from \cite{arad_improved_2012,arad_area_2013}:
\begin{equation}\label{toself}\ket z\in\GS\implies\AGSP\ket z\in\GS\aand\ket y\in\GS^\perp\implies\AGSP\ket y\in\GS^\perp.\end{equation}
Indeed, these imply $\AGSP\P_\GS=\P_\GS\AGSP\P_\GS=\P_\GS\AGSP-\P_\GS\AGSP\P_{\GS^\perp}=\P_\GS\AGSP$, where the last equality is because $\AGSP$ sends $\GS^\perp$ to itself. In the special case where $\AGSP$ is Hermitian it suffices to check one of the implications \eqref{toself}.  

\begin{observation}\label{directsum}
$\AGSP$ is a $\shrink$-AGSP for $\Z\subsp\H$ iff it is of the form $\agsp_{\Z}\oplus\agsp_{\Z^\perp}$ where $\agsp_{\Z}\in\B(\Z)$ is a dilation (in particular $\agsp_\Z$ is invertible), and $\agsp_{\Z^\perp}\in\B(\Z^\perp)$ satisfies $\|\agsp_{\Z^\perp}\|\le\sqrt\shrink$.
\end{observation}

Letting $\Z$ be the lowest-energy eigenspace of some Hamiltonian we can compare refinition \ref{AGSPdef} with standard definitions of AGSPs.
\begin{observation}
Definition \ref{AGSPdef} includes the definitions of AGSP in \cite{arad_improved_2012,arad_area_2013} and of \emp{spectral} AGSP \cite{arad_rigorous_2017} as special cases.\end{observation}
Specializing item \ref{dilation} of definition \ref{AGSPdef} to $\AGSP\P_\GS=\P_\GS$ one obtains the definition of AGSP in \cite{arad_improved_2012,arad_area_2013}. A \emph{spectral} AGSP \cite{arad_rigorous_2017} corresponds to definition \ref{AGSPdef} with the additional requirement that $\AGSP\opge0$, and that $\AGSP$ and $H$ be simultaneously diagonalizable.

\subsection{Comparing subspaces}
We use the terminology of \emph{overlap} $\mu$ and \emph{viability} error $\delta$ of \cite{arad_rigorous_2017}:
\begin{definition}Given subspaces $\Z,\V\subsp\H$, the \emp{overlap} of $\V$ onto $\Z$ is $\onto\V\Z=\min_{\ket z\in\sphere(\Z)}{\bra{z}\P_\V\ket{z}}$. Letting $\mu=\onto\V\Z$, say that $\delta=1-\mu$ is the \emp{viability error} of $\V$ onto $\Z$. $\erra=\delta/\mu$ is the \emp{error ratio}.
\end{definition}
\begin{definition}
	Given subspaces $\Z,\V\subsp\H$ we say that $\V$ \emp{covers} $\Z$ if $\onto\V\Z>0$, or equivalently $\P_\Z(\V)=\Z$. Introduce notation $\supsp_\mu$ and $\parallel_\mu$ as follows:
	\begin{itemize}
		\item\textup{(${\V\supsp_\mu\Z}$)}``$\V$ is {$\mu$-overlapping} onto $\Z$''  if $\onto\V\Z\ge\mu$. 
		\item$({\V\parallel_\mu\Z})$	``$\V$ and $\Z$ are {mutually $\mu$-overlapping}''  if ${\V\supsp_\mu\Z}$ and ${\Z\supsp_\mu\V}$.
	\end{itemize}
\end{definition}

\subsection{The transition map and symmetry between subspaces}

Let $\ON\V$ denote the orthogonal projection on a subspace $\V\subsp\H$ \emph{ when viewed as a map $\H\to\V$} (i.e., with restricted codomain as opposed to $\P_\V:\H\to\H$. Note that $\P_\V=\IN\V\ON\V$).

Given another subspace $\Z\subsp\H$ we define the \emph{transition map $\tofrom\V\Z$ from $\Z$ to $\V$} as the restriction of $\ON\V$ to domain $\Z$. Formally we have:
\begin{definition}
Given a subspace $\V\subsp\H$,  let $\ON\V:\H\to\V$ be the adjoint of the inclusion map $\IN\V:\V\into\H$.
The transition map from $\Z$ to $\V$ is $\tofrom\V\Z=\ON\V\IN\Z$.
\end{definition}
The overlap $\mu$ of $\V$ onto $\Z$ equals $\onto{\V}{\Z}=\min\spec(\tofrom\Z\V \tofrom\V\Z)$, where $\spec$ is the spectrum.  
The \emph{principal angles between $\Z$ and $\V$} are defined \cite{galantai_jordans_2006,ben-israel_geometry_1967} as the $\arccos$ of the singular values of $\tofrom\V\Z$.
This definition illustrates a symmetry between two subspaces. Exploiting this symmetry is essential to us in proving the sharp error reduction.


\begin{observation}\label{same}
If $\onto\Z\V,\onto\V\Z>0$ then $\onto\V\Z=\onto\Z\V$.
\end{observation}
\begin{proof}
Let $M=\tofrom\Z\V$. Then $\spec(M M^\dag)\xpt\{0\}=\spec(M^\dag M)\xpt\{0\}$ (Jacobson's lemma). The assumed nonzero overlaps then imply $\spec(M^\dag M)=\spec(MM^\dag)$. So $\onto\V\Z=\min\spec(MM^\dag)=\min\spec(M^\dag M)=\onto\Z\V$.\end{proof}

\begin{corollary}[Symmetry lemma \cite{abrahamsen_polynomial-time_2019}]\label{newswap}
For $\V_1,\V_2\subsp\H$ and $\mu>0$,
\[\V_1\subsp_\mu\V_2\aand\text{$\V_1$ covers $\V_2$}\iiff\V_1\parallel_\mu\V_2.\]
\end{corollary}

\begin{lemma}\label{obv}
For subspaces $\Z,\Y\subsp\H$ and $\mu>0$,
\[\V\supsp_\mu\Z\iiff\P_\V(\Z)\parallel_\mu\Z.\]
\end{lemma}
\begin{proof}
	($\Leftarrow$) is clear. ($\Rightarrow$): Let $\Y=\P_\V(\Z)$. Then $\P_\Z\P_\V\P_\Z=\P_\Z\P_\Y\P_\Z$, so $\V\supsp_\mu\Z\implies\Y\supsp_\mu\Z$. But $\P_\Y(\Z)=\P_\Y\P_\V(\Z)=\P_\Y(\Y)=\Y$ so $\Z$ covers $\Y$, hence $\Y\parallel_\mu\Z$ by the symmetry lemma.
\end{proof}

\section{Proof of sharp error reduction}\label{thesharpproof}
The proof of lemma \ref{sharplemma} uses the {symmetry lemma} to switch the roles of the approximated subspace $\GS$ and the approximating subspace (first $\V$ then $\V'$).
\begin{lemma*}Let $\AGSP$ be a $\shrink$-AGSP for $\GS\subsp\H$, and suppose $\V\subsp\H$ covers $\GS$ with error ratio $\erra$. then $\V'=\AGSP(\V)$ covers $\GS$ with error ratio $\erra'\le\shrink\cdot\erra$.
\end{lemma*}
\begin{proof}
	By the symmetry lemma it suffices to find a subspace $\Y'\subsp\V'$ such that
	\begin{center}
		\begin{enumerate*}[label=\textbf{(\roman*)}]
		\item $\Y'$ covers $\Z$, and\qquad \label{Y_overlaps_Z}
		\item $\Z$ covers $\Y'$ with error ratio $\shrink\cdot\erra$.\label{Z_mu_Y}
	\end{enumerate*}
\end{center}
	Indeed, this implies $\Z\parallel_{\mu'}\Y'$ with $\mu'=\frac1{1+\shrink\cdot\erra}$. Thus $\Z\subsp_{\mu'}\V'$ and we are done.

	\paragraph{Pick $\bs{\Y'=\AGSP\P_\V(\Z)}$.} We argue items \ref{Y_overlaps_Z} and \ref{Z_mu_Y}:
\paragraph{\ref{Y_overlaps_Z}} $\P_\Z(\Y')=\AGSP\P_\Z\P_\V\P_\Z(\Z)=\AGSP\Z=\Z$. This shows that $\Y'$ covers $\Z$.

\paragraph{\ref{Z_mu_Y}} (show $\Z\supsp_{\mu'}\Y'$) 
It suffices to show that, given an arbitrary $\ket{y'}\in\Y'$,
		\begin{equation}\label{ypPyp}{\bra{y'}\P_{\Z^\perp}\ket{y'}}\le\shrink\erra{\bra{y'}\P_\Z\ket{y'}}.\end{equation}
		Let $\Y=\P_\V(\Z)$ and pick $\ket y\in\Y$ such that $\ket{y'}=\AGSP\ket y$ (since $\Y'=\AGSP\Y$).
By lemma \ref{obv} we have $\Y\parallel_\mu\Z$ where $\mu=\frac1{1+\erra}$. In particular $\Z\supsp_\mu\Y$ implies
${\bra{y}\P_{\Z^\perp}\ket{y}}\le\erra{\bra{y}\P_\Z\ket{y}}.$
Apply the bound $\|\AGSP\P_{\Z^\perp}\|\le\sqrt\shrink$ and the dilation property on $\Z$:
\begin{equation}\|\AGSP\P_{\Z^\perp}\ket{y}\|\le\sqrt\shrink\|\P_{\Z^\perp}\ket{y}\|\le\sqrt{\shrink\erra}\|\P_\Z\ket{y}\|\le\sqrt{\shrink\erra}\|\AGSP\P_\Z\ket y\|.\end{equation}
Recognizing the LHS as $\|\P_{\Z^\perp}\ket{y'}\|$ and the RHS as $\sqrt{\shrink\erra}\|\P_\Z\ket{y'}\|$ establishes \eqref{ypPyp}.
\end{proof}

\section{Degenerate-case entanglement bound} \label{EB}
With the sharp error reduction bound in hand the degenerate-case OTR entanglement bound follows using a standard \emph{bootstrapping argument} \cite{arad_improved_2012,arad_area_2013,arad_rigorous_2017}. This argument combines a method to increase overlap with one to reduce entanglement. %
The entanglement is represented in terms of the {dimension} of a $\mu$-overlapping subspace on the left subsystem.

\subsection{Preliminaries on bipartite spaces}\label{prelimbipart}

\begin{definition}
	Given a subspace $\GS\subsp\Hl\otimes\Hr$ of a bipartite space and a subspace $\V\subsp\Hl$ of the left tensor factor, define the (left) overlap $\mu$, viability error $\delta$, and error ratio $\erra=\delta/\mu$ of $\V$ onto $\GS$ as the corresponding parameters for $\V\otimes\Hr$ onto $\GS$.
\end{definition}

\begin{definition}[\cite{abrahamsen_polynomial-time_2019}]
A $\shrink$-PAP (partial approximate projector) $\PAP$ with target space $\GS\subsp\Hl\otimes\Hr$ is a space of operators $\PAP\subsp\Bl$ such that $\PAP\otimes\Br$ contains some $\shrink$-AGSP $\AGSP$ for $\GS$.
\end{definition}
If $\AGSP\in\B(\Hl\otimes\Hr)$ is a $(\shrink,R)$-AGSP, then there exists a corresponding $\shrink$-PAP $\PAP\subsp\Bl$ with $\vdim\PAP\le R$.
\begin{corollary}[of lemma \ref{sharplemma}]
Let $\PAP\subsp\Bl$ be a $\shrink$-PAP with target space $\GS\subsp\Hl\otimes\Hr$, and suppose $\V\subsp\Hl$ covers $\GS$ with error ratio $\erra$. Then $\PAP\V=\{\agsp\ket v:\agsp\in\PAP,\ket v\in\V\}$ covers $\GS$ with error ratio $\erra'\le\shrink\cdot\erra$.
\end{corollary}

\subsection{Applying the bootstrap \cite{arad_improved_2012,arad_area_2013,arad_rigorous_2017}}
The {bootstrapping argument} \cite{arad_improved_2012,arad_area_2013,arad_rigorous_2017} proves the existence of a subspace $\V\subsp\Hl$ with small dimension and non-negligible overlap with the target space $\GS\subsp\Hl\otimes\Hr$. The argument combines a method for reducing the entanglement of a subspace with one for increasing overlap with the target space (i.e., an AGSP) in such a way that $\vdim\V$ does not increase when concatenating the operations. 

To offset the dimension growth from the AGSP, the entanglement reduction needs to decrease the entanglement by an factor $R$, which means decreasing the overlap by a factor $\Theta(R)$ using the dimension reduction procedure of \cite{arad_rigorous_2017} (appendix \ref{dimred}). One therefore has to apply the $(\shrink,R)$-AGSP in the low-overlap regime $\mu=c/R$. If we used the error bound $\delta'=\shrink/\mu^2$ of \cite{arad_rigorous_2017} then we would need $\shrink<\mu^2=(c/R)^{2}$ to have any bound on the post-AGSP error, hence requiring a bound of the form $R^2\shrink<\tilde c$ on the parameter tradeoff for the AGSP. In contrast, lemma \ref{sharplemma} weakens this requirement to $\shrink=\mu=cR$. More precisely we will use:
\begin{corollary}\label{sharpapplication}
	Let $\AGSP$ be a $\shrink$-AGSP with target space $\GS\subsp\H$, and suppose $\V\subsp\H$ $\mu$-overlaps onto $\GS$ with $\mu\ge\shrink$. Then $\V'=\AGSP(\V)$ has overlap $\mu'=1/2$ onto $\GS$.
\end{corollary}
\begin{proof}
	$\V$ has error ratio $\erra=\frac{1-\mu}\mu\le\frac1\mu$, so $\V'$ has error ratio $\erra'\le\shrink/\mu\le1$ by lemma \ref{sharplemma}. This corresponds to overlap $\mu'=\frac1{\erra'+1}\ge1/2$.
\end{proof}

The following lemma is proven following the overall argument of \cite{arad_rigorous_2017} proposition 2 and combining it with the sharp error reduction bound in the form of corollary \ref{sharpapplication} to change the condition from a bound on $R^C\shrink$ to one on $R\shrink$. In the following $x\lesssim y$ means $x=O(y\vee1)$ where $\vee$ denotes the maximum.
\begin{lemma}\label{start}
	Let $\GS\subsp\Hl\otimes\Hr$ be a subspace with degeneracy $\vdim\GS=D$. If there exists a $(\shrink,R)$-AGSP $\AGSP\in\B(\Hl)\otimes\B(\Hr)$ with target space $\GS$ and parameters such that
	\begin{equation}\shrink\cdot R\le1/32,\label{shrinkmu}\end{equation}
then there exists a left $\frac1{32R}$-overlapping space $\V\subsp\Hl$ onto $\GS$ such that $\vdim\V\plusbounded D\log R$. 
It follows that there exists $\V''$ of dimension $\vdim{\V''}\lesssim DR^2\log R$ which is left $\shrink$-viable for $\GS$.

\end{lemma}
\begin{proof}
	Let $\V$ be a left $\nu=\frac1{32R}$-overlapping space onto $\GS$ whose dimension $V$ is minimal with respect to this property. 
Let $\PAP$ be the $\shrink$-PAP of dimension $R$ associated to the $(\shrink,R)$-AGSP $\AGSP$, and let $\V'=\PAP\V$ so that $V'=\vdim{\V'}\le RV$. $\shrink\le\nu$ by assumption \eqref{shrinkmu}, so corollary \ref{sharpapplication} yields that $\V'$ is $1/2$-overlapping onto $\GS$. 

By corollary \ref{probmeth} there exists $\Y'\subsp\V'$ which is left $\nu=\frac1{32R}$-overlapping onto $\GS$ and has dimension at most $V/2+O(D\log R\vee\log V)$ since $8V'\cdot\frac{1/(32R)}{1/2}\le V/2$. By minimality of $\V$ we have that $V\le V/2+O(D\log R\vee\log V)$, and rearranging yields the result about $\V$.

The last remark follows by taking $\V''=\PAP^2\V=\PAP\V'$. Then $\V''$ covers $\GS$ with error ratio $\erra''\le\shrink$ by lemma \ref{sharplemma} since $\V'$ has $\erra'=1$, and this upper-bounds the viability error.
\end{proof}

\subsection{Subspace overlap $\to$ entanglement of vectors}

The following lemma relates the entanglement of individual ground states to $\delta$-viability. 
\begin{lemma}\label{tailbound}
	Let $\GS\subsp\Hl\otimes\Hr$ and suppose there exists a $\delta$-viable space $\mc V\subset\Hl$ of dimension $V$ for $\GS$. Pick any state $\ket\psi\in\sphere(\GS)$ and write the Schmidt decomposition $\sum_i\sqrt{\lambda_i}\ket{x_i}\ket{y_x}\in\sphere(\GS)$ with non-increasing coefficients. Then we have the tail bound
	\[\sum_{i=V+1}^{\vdim\Hl}\lambda_i\le\sqrt\delta.\]
\end{lemma}
\begin{proof}
	Let $\ket\phi\in\sphere(\GS)$ such that $\bracket{\psi}{\phi}^2\ge1-\delta$, and let $\rho_\psi$ and $\rho_\phi$ be the reduced density matrices on $\Hl$ so that $\lambda_i=\lambda_i^\psi$ are the eigenvalues of $\rho_\psi$. Then, since the trace distance contracts under the partial trace:
	\[\frac12\|\rho_\psi-\rho_\phi\|_1\le\frac12\big\|\ket\psi\bra\psi-\ket\phi\bra\phi\big\|_1=\sqrt{1-\bracket{\psi}{\phi}^2}\le\sqrt\delta,\]
	Let $d\rho=\rho_\psi-\rho_\phi$ and call its non-increasing eigenvalues (not all positive) $\lambda_i^{d\rho}$ and let $\lambda^\phi_i$ be the non-increasing eigenvalues of $\phi$. For $V+1\le i\le\vdim\Hl$, Weyl's inequalities imply $\lambda_{i}^\psi\le\lambda^\phi_{V+1}+\lambda_{i-V}^{d\rho}=\lambda_{i-V}^{d\rho}$.
	Thus $\sum_{i>V}\lambda_{i}\le\sum_j(\lambda_j^{d\rho})_+=\frac12\|d\rho\|_1\le\sqrt\delta$ where $(x)_+=x\vee0$ is the positive part and the middle equality is because $\opn{tr}(d\rho)=0$.
\end{proof}

\subsection{Proof of proposition \ref{mainres}}\label{frsection}

In the case of a frustrated Hamiltonian the AGSP contruction involves a spectral \emph{truncation} of the Hamiltonian on either side of a cut, incurring an error in the target space of the AGSP. We first prove a version of proposition \ref{mainres} which is applicable to the frustrated case by allowing the target space to be approximate. We then specialize to the case of an exact target space to obtain proposition \ref{mainres}.

Say that subspaces $\tilde\GS,\GS\subsp\H$ are $\delta$-close ($\tilde\Z\approx_\delta\Z$) if $\tilde\Z\parallel_{1-\delta}\Z$.

\begin{lemma}\label{frversion}
	Let $\GS$ with degeneracy $\vdim\GS=D$ be a subspace of bipartite space $\H=\Hl\otimes\Hr$. Let $\tilde\GS_1,\tilde\GS_2,\ldots\subsp\H$ be a sequence of subspaces such that $\tilde\GS_n\approx_{\delta_n}\GS$ where $\delta_1,\delta_2,\ldots$ is a sequence such that $\sum_{n=0}^\infty n\sqrt{\delta_n}=O(1)$.

Let $R\shrink\le1/2$ and suppose there exists a sequence $\AGSP_1,\AGSP_2,\ldots$ such that $\AGSP_n$ is an $(\shrink^n,R^n)$-AGSP for target space $\tilde\Z_n$. Then,
\begin{equation}\max_{\ket\psi\in\sphere(\GS)}\shann(\rho^\psi_\lft)\le(1.01+c_\delta)\log D+O(\log R)\quad\text{where}\quad c_\delta=\sum_{n=1}^\infty\sqrt{\delta_n}.\label{simplified}\end{equation}
\end{lemma}
\begin{proof}
For any $m=5,6\ldots$ we show that $S(\rho^\psi_{\Hl})$ is bounded by
\begin{equation}\label{eqm}(1+\epsilon_m+c_\delta)\log D+O(m\log R),\quad\text{where}\quad\epsilon_m=\frac{\shrink^{m/2}}{1-\shrink^{1/2}}.\end{equation}
\eqref{simplified} then follows by taking $m=17$ since that and $\shrink\le1/2$ yield $\epsilon_n\le0.01$.

Applying lemma \ref{start} to $\PAP_n$ yields a left $\shrink^n$-viable space for $\tilde\GS$ for each $n\ge m$ since $R^n\shrink^n\le1/32$. The lemma implies that $\vdim{\V_n}\lesssim DR^{2n}\log(R^n)$ and hence $\vdim{\V_n}\le CDR^{3n}$ for a constant $C>0$.
	$\V_n$ is $(\shrink^{n/2}+\delta^{n/2})^2$-viable for $\GS$ by the proof of \cite{arad_rigorous_2017} lemma 3.
By lemma \ref{tailbound} the Schmidt coefficients of any state $\ket\psi\in\sphere(\GS)$ satisfy $\sum_{i>C DR^{3n}}\lambda_i\le\shrink^{\frac{n}2}+\sqrt{\delta_n}$ for each $n\ge5$. 

Let $I_0=\{1,2,\ldots,CD\cdot R^{3m}\}$, $I_1,\ldots,I_{m-1}=\emptyset$, and $I_n=\NN\cap(CD\cdot R^{3n},CD\cdot R^{3(n+1)}]$ for $n\ge m$. By the standard decomposition \cite{arad_improved_2012} of the Shannon entropy described in lemma \ref{shannonentropy} (appendix \ref{thestandarddec}),
\begin{align*}\shann(\Lambda_i)&\le\log(C D R^{3m})+\sum_{n=m}^\infty(\shrink^{n/2}+\sqrt{\delta_n})\log(CDR^{3n+3})+\sum_{n=m}^\infty h(\shrink^{n/2}+\sqrt{\delta_n})\\&=(1+\epsilon_m)\log D+O(m\log R)+\sum_{n=m}^\infty h(\shrink^{n/2}+\sqrt{\delta_n}).\end{align*}
We finalize by bounding the rightmost sum. Since $h$ is increasing on $[0,1/e]$, we can bound $h(\shrink^{n/2}+\sqrt{\delta_n})$ by $h(2^{-\frac n2}+\sqrt{\delta_n})$. This in turn is bounded by
\[h(2^{-\frac n2}+\sqrt{\delta_n})\le(2^{-\frac n2}+\sqrt{\delta_n})\log(2^{\frac n2})\le h(2^{-\frac n2})+ n\sqrt{\delta_n}.\]
So $\sum_n h(\shrink^{n/2}+\sqrt{\delta_n})=O(1)$. This establishes \eqref{eqm}.
\end{proof}
The coefficient $1.01$ in lemma \ref{frversion} and proposition \ref{mainres} can be replaced by $1+\epsilon$ for any fixed $\epsilon>0$ by taking $m\propto\log(1/\epsilon)$. The implicit constant of $O(\log R)$ then depends logarithmically on $1/\epsilon$.
\begin{proof}[Proof of proposition \ref{mainres}]
Given AGSP $\AGSP$ with $R\shrink\le1/2$ apply lemma \ref{frversion} to the sequence of AGSPs $\AGSP_n=\AGSP^n$, each with the exact target space $\tilde\GS_n=\GS$ such that we can take $\delta_n=0$. \end{proof}

\newpage
\section{Alternative proof of sharp error reduction}
\label{liftingversion}


Let $\Z,\V\subsp\H$ be subspaces such that 
$\P_\Z\P_\V\P_\Z\opge\mu\P_\Z$ (i.e, $\V\supsp_\mu\Z$). Lemmas 1 and 2 of \cite{arad_rigorous_2017} state that for every $\ket z\in\Z$ there exists $\ket v\in\V$ with norm at most $\|v\|\le\mu^{-1}\|z\|$ such that $\P_\Z\ket v=\ket z$. The alternative proof of the error reduction lemma \ref{sharplemma} relies on noticing that this statement can be improved quadratically, i.e., we can replace $\mu^{-1}$ with $\mu^{-1/2}$.
\subsection{Quadratically improved lifting lemma}
\begin{definition}
Let $\Z,\V\subsp\H$ be subspaces such that $\V$ covers $\Z$. Define the \emp{lifting operator} from $\Z$ to $\V$ as $\liftofrom\V\Z=\tofrom\V\Z(\ON\Z\P_\V\IN\Z)^{-1}$.
\end{definition}
\begin{lemma}\label{liftlemma}
Given subspaces $\Z,\V\subsp\H$ such that $\V\supsp_\mu\Z$ with $\mu>0$, the lifting operator $\Z\to\V$ satisfies:
\begin{enumerate}
	\item $\P_\Z\circ\liftofrom\V\Z\ket z=\ket z$ for any $\ket z\in\Z$\qquad \hfill\textup{(lifting property)},
		\item \label{opnorm}$\|\liftofrom\V\Z\|\le\mu^{-1/2}$,
	\end{enumerate}
\end{lemma}
\begin{proof}
	The restricted projection $M=\tofrom\Z\V$ is surjective since $\V$ covers $\Z$, so $M^\dag(MM^\dag)^{-1}=\liftofrom\V\Z$ is a well-defined right-inverse\footnote{This right-inverse is a special case of the Moore-Penrose pseudo-inverse, but its role is not analogous to the pseudoinverse in the proof of \cite{arad_rigorous_2017} lemma 6, which was a pseudoinverse of the \emph{AGSP}.} of $M$. This is the lifting property. For the norm bound we write the polar decomposition $\tofrom\Z\V=SV^\dag$ where $S$ is a positive operator on $\Z$ and $V^\dag$ is the adjoint of an isometry $V:\Z\to\V$ (again using that $M$ is surjective). Since $\V\supsp_\mu\Z$ we have $\mu\id_\Z\ople MM^\dag= SV^\dag VS=S^2$ which implies that $S\opge\sqrt\mu\id_\Z$. Then $\|\liftofrom\V\Z\|=\|VS^{-1}\|\le\mu^{-1/2}$.
\end{proof}

We also write $\liftofrom\V\Z$ in the same way when extending its codomain and viewing it as a map $\Z\to\H$.
By Pythagoras' theorem, $\|z\|^2+\|\P_{\Z^\perp}\liftofrom\V\Z\ket z\|^2=\|\liftofrom\V\Z\ket z\|^2\le\mu^{-1}\|z\|^2$.  Since $\erra=\mu^{-1}-1$, rearranging yields:
\begin{corollary}\label{heightcor}
	Let $\Z,\V\subsp\H$ be such that $\V$ covers $\Z$ with error ratio $\erra$. Then for any $\ket z\in\Z$,
	\begin{equation*}\label{perpdec}\liftofrom\V\Z\ket z=\ket z+\P_{\Z^\perp}\liftofrom\V\Z\ket z\qquad\text{where}\qquad\|\P_{\Z^\perp}\liftofrom\V\Z\|\le\sqrt\erra.\end{equation*}
\end{corollary}

The alternative proof of lemma \ref{sharplemma} can now be finalized essentially as in the proof of \cite{arad_rigorous_2017} lemma 6:
\begin{proof}[Finishing the alternative proof of lemma \ref{sharplemma}]
	Write the AGSP as $\agsp_\Z\oplus\agsp_{\Z^\perp}$ as in observation \ref{directsum}. 
	Given an arbitrary unit vector $\ket z\in\Z$ pick $\ket{v'}=\AGSP\circ\liftofrom\V\Z\circ\agsp_\Z^{-1}\ket{z}\in\AGSP\V$. It suffices to show that $\bracket{{z}}{v'}^2\ge\mu'\|v'\|^2$ where $\mu'=\frac1{1+\shrink\erra}$: Applying corollary \ref{heightcor} to $\agsp_\Z^{-1}\ket z$ we have the orthogonal decomposition
	\[\ket{v'}=\ket{z}+\ket{h'}\quad\text{where}\quad\ket{h'}=\agsp_{\Z^\perp}(\P_{\Z^\perp}\liftofrom\V\Z)\agsp_\Z^{-1}\ket{z}.\]
where $\|h'\|\le\|\agsp_{\Z^\perp}\|\cdot\|\P_{\Z^\perp}\liftofrom\V\Z\|\le\sqrt{\shrink\erra}$.  
Then $\|v'\|^2\le1+\shrink\erra$ by Pythagoras', so $\bracket{{z}}{v'}^2/\|v'\|^2=1/\|v'\|^2\ge\mu'$.
\end{proof}


\section{Acknowledgements}
The author thanks Peter Shor and Jonathan Kelner for inspiring discussions, and he thanks Anurag Anshu and David Gosset for insightful dicussions about their proof \cite{anshu_entanglement_2019} of a 2D subvolume law.

\begin{appendices}
	\section{}
\subsection*{The AGSP of \cite{anshu_entanglement_2019} has the same parameters in the degenerate case}\label{2d_gen}
The spectral bound of the AGSP is based on the coarse-grained DL operator \cite{aharonov_detectability_2009,arad_improved_2012} (into bands $H_i$ of width $4t=\tilde O(L^{1/3}\gamma^{-2/3})$) and its analysis in \cite{anshu_simple_2016} which is explicitly independent of degeneracy. A modification makes the overlaps of the coarse-grained projectors smaller by a factor 2. This modification was justified in lemma 3.1, a reduction which is also valid in the degenerate case. \cite{anshu_entanglement_2019} then replaced the product of a set of $m=\tilde O(L^{1/3}\gamma^{-1/6})$ coarse-grained projectors $Q_i$ near the cut by a polynomial in the $m$ Hamiltonians $H_i$ corresponding to the $m$ disjoint bands, and the analysis bounds the approximation error in operator norm. This analysis is made in an eigenbasis for $H_i$ on each band and is agnostic to the global ground state (or space). Hence the shrinking factor is independent of the degeneracy. As for the entanglement rank of the AGSP it is bounded by a combinatorial argument (theorem 5.1) which does not depend on the spectral properties of $H$, so this is again valid for degenerate Hamiltonians.

\section{}

\subsection{Standard entropy bound \cite{arad_improved_2012} from partial sums}
\label{thestandarddec}
A bound on the Shannon entropy of a probability distribution can be obtained through a dyadic decomposition by following the argument of \cite{arad_improved_2012} lemma III.3. Given a sequence $\Lambda=(\lambda_1,\lambda_2,\ldots)\in[0,1]^{\NN}$ write the Shannon entropy $\shann(\Lambda)=\sum_ih(\lambda_i)$ where $h(x)=x\log(x^{-1})$.
\begin{claim}[\cite{arad_improved_2012,arad_area_2013}]\label{shannonentropy}
Let $\Lambda=(\lambda_1,\lambda_2,\ldots)\in[0,1]^{\NN}$ be a sequence with $\sum_i\lambda_i\le 1$ and write $\Sigma_I=\sum_{i\in I}\lambda_i\le1$ for $I\subset\NN$. Let $I_0,I_1,\ldots$ be a partition on $\NN$ such that $\Sigma_{I_n}\le \gamma_n$ for some sequence of $\gamma_n\in[0,1]$.  If $|I_n|\ge3$ for each $n$, then
\[\shann(\Lambda_i)\le\log|I_0|+\sum_{n=1}^\infty\gamma_n\log(|I_n|)+\sum_{i=1}^\infty h(\gamma_n).\] 
\end{claim}
\begin{proof}
	Since $h$ is concave Jensen's inequality states that for any set of indices $I$, $\frac1{|I|}\sum_{i\in I}h(\lambda_i)\le h(\tfrac1{|I|}\sum_{i\in I}\lambda_i)$. Rearranging yields:
	\begin{equation}\label{Jensens}\sum_{i\in I}h(\lambda_i)\le |I|\cdot h(\nofrac{\Sigma_I}{|I|}).\end{equation}
	$h$ is increasing on $[0,1/e]$, so if $|I|\ge3$ and $\gamma\le1$ is an upper bound on $\Sigma_I$, then $\sum_{i\in I}h(\lambda_i)\le |I|h(\gamma/|I|)=\gamma\log(|I|\gamma^{-1})$. Apply this bound for each $n=1,2,\ldots$. We also have in particular that $\sum_{i\in I}h(\lambda_i)\le\log|I|$. Apply this for $I_0$.
\end{proof}

\subsection{Dimension reduction}\label{dimred}
Having analyzed the AGSP which achieves the improvement of the overlap we now recall a standard tool for entanglement reduction. 
\begin{lemma}[\cite{arad_rigorous_2017}]\label{randlem}
Let $\GS\subsp\Hl\otimes\Hr$ be a subspace with dimension $D$ and let $\W\subsp\Hl$ be left $\mu$-overlapping onto $\GS$ with $\vdim\W=W$. Then a Haar-uniformly random subspace $\V\subsp\W$ of dimension $V\le W$ is left $\nu$-overlapping onto $\GS$ with probability at least $1-\eta$ where
\[\nu=\frac{V}{8W}\cdot\mu\aand\eta=(1+2\nu^{-1/2})^D We^{-V/16}.\]
\end{lemma}
Since $1+2x\le 3x$ for $x>1$ (and in particular for $x=\nu^{-1/2}\ge\sqrt8$) we have the bound on the error probability:
\begin{equation}\eta<(9/\nu)^{D/2}We^{-V/16}.\label{ttbound}\end{equation}
Applying the probabilistic method we obtain:
\begin{corollary}\label{probmeth}
Let $\W\subsp\Hl$ of dimension $W$ be left $\mu$-overlapping onto $\GS\subsp\Hl\otimes\Hr$ with $\vdim\GS=D$. For any $0<\nu\le\mu$ there exists a subspace $\V\subsp\W$ which is left $\nu$-overlapping onto $\GS$ and has dimension
\begin{equation}\label{thisV}V=\Big\lceil 8\Big(W\cdot\frac\nu\mu\vee\big(D\log(9/\nu)+2\log W\big)\Big)\Big\rceil\wedge W.\end{equation}
\end{corollary}
\begin{proof}
	If $V=W$ then $\V=\W$ suffices. Otherwise let $\tilde\nu=\frac{V}{8W}\mu$ be the overlap from lemma \ref{randlem} corresponding to the choice \eqref{thisV} of $V$ and let $\tilde\eta=(9/\nu)^{D/2}We^{-V/16}$. Then $\log \tilde\eta=\tfrac D2\log(9/\nu)+\log W-V/16\le0$ by the choice of $V$.  
	By \eqref{ttbound} the error probability in lemma \ref{randlem} is strictly below $\tilde\eta\le1$ so by the probabilistic method there exists a left $\tilde\nu$-overlapping space. But $\tilde\nu\ge\nu$ which proves the claim.
\end{proof}

	\end{appendices}
\bibliographystyle{alpha}
\bibliography{../bibliography}

\end{document}